\documentclass[fleqn]{llncs}
\usepackage{pgaramp}
\usepackage{hyperref}

\newcommand{\doi}[1]{{doi:\href{http://doi.org/#1}{\nolinkurl{#1}}}}

\newcommand{\eprintlnk}[1]{{\href{#1}{Electronic version}}}
\renewcommand{\url}[1]{{\href{#1}{\nolinkurl{#1}}}}

\title{Program Algebra for \\ Random Access Machine Programs}
\author{C.A. Middelburg \\
        {\small ORCID: \url{https://orcid.org/0000-0002-8725-0197}}}

\institute{Informatics Institute, Faculty of Science,
           University of Amsterdam, \\
           Science Park~900, 1098~XH Amsterdam, the Netherlands \\
           \email{C.A.Middelburg@uva.nl}}

\begin{document}

\maketitle

\begin{abstract}
This paper presents an algebraic theory of instruction sequences with
instructions for a random access machine (RAM) as basic instructions, 
the behaviours produced by the instruction sequences concerned under 
execution, and the interaction between such behaviours and RAM memories.
This theory provides a setting for the development of theory in areas 
such as computational complexity and analysis of algorithms that 
distinguishes itself by offering the possibility of equational reasoning
to establish whether an instruction sequence computes a given function 
and being more general than the setting provided by any known version of 
the RAM model of computation.
In this setting, a semi-realistic version of the RAM model of 
computation and a bit-oriented time complexity measure for this version 
are introduced. 
Under the time measure concerned, semi-realistic RAMs can be simulated 
by multi-tape Turing machines with quadratic time overhead.
\begin{classcode}
D.1.4, E.3, F.1.1, F.1.3.
\end{classcode}
\end{abstract}

\section{Introduction}
\label{sect-intro}

This paper introduces an algebraic theory which provides a setting for 
the development of theory in areas such as computational complexity and 
analysis of algorithms that elaborates on a version of the random access 
machine (RAM) model of computation.
The setting in question distinguishes itself by offering the possibility 
of equational reasoning to establish whether an instruction sequence 
computes a given function, and by being more general than the setting 
provided by any known version of the RAM model of computation.
Many known and unknown versions of this model of computation can be 
dealt with by imposing apposite restrictions.
We expect that the generality is conducive to the investigation of novel 
issues in the areas of computational complexity and analysis of 
algorithms.
This expectation is based on our experience with a comparable algebraic 
theory of instruction sequences, where instructions operate on Boolean 
registers, in previous work 
(see~\cite{BM13b,BM13a,BM14a,BM14e,BM13c,BM18a}).

This paper belongs to a line of research that started with~\cite{BL02a}, 
and of which an enumeration is available at~\cite{SiteIS}. 
The first objective of this line of research is to understand the 
concept of a program. 
The notion of an instruction sequence appears in the work in question as 
a mathematical abstraction for which the rationale is based on this 
objective. 
The structure of the mathematical abstraction at issue has been 
determined in advance with the hope of applying it in diverse 
circumstances where in each case the fit may be less than perfect. 
Until now, this work has, among other things, yielded an approach to 
computational complexity where program size is used as complexity 
measure, a contribution to the conceptual analysis of the notion of an 
algorithm, and new insights into such diverse issues as the halting 
problem, garbage collection, program parallelization for the purpose of 
explicit multi-threading, and virus detection.

The basis of all the work in question (see~\cite{SiteIS}) is the 
combination of an algebraic theory of single-pass instruction sequences, 
called program algebra, and an algebraic theory of mathematical objects 
that represent the behaviours produced by such instruction sequences 
under execution, called basic thread algebra, extended to deal with the 
interaction between such behaviours and components of an execution 
environment for instruction sequences.
This combination is parameterized by a set of basic instructions and a 
set of mathematical objects that represent the execution environment 
components.
                  
The current paper contains a simplified presentation of the 
instantiation of this combination in which RAM memories are taken as the 
components of an execution environment, instructions for a RAM are taken 
as basic instructions, and an execution environment consists of only one 
component. 
Because we opt for the most general instantiation, all instructions 
that do not read out or alter more than one register from the RAM memory 
are taken as basic instructions.
Both known and unknown versions of the RAM model of computation can be
dealt with by restriction on the set of basic instructions.
We expect that by this set-up the presented instantiation can be useful 
to rigorous investigations of novel issues relating to computational 
complexity and analysis of algorithms.

Program algebra and basic thread algebra were first presented 
in~\cite{BL02a}.%
\footnote
{In that paper and the first subsequent papers, basic thread algebra 
 was introduced under the name basic polarized process algebra.}
The extension of basic thread algebra referred to above, an extension to 
deal with the interaction between the behaviours produced by instruction 
sequences under execution and components of an execution environment, 
was first presented in~\cite{BM09k}.
The presentation of the extension is rather involved because it is
parameterized and owing to this covers a generic set of basic 
instructions and a generic set of objects that represent execution 
environment components.
In the current paper, a much less involved presentation is obtained by 
considering only the case where execution environment components are 
RAM memories, basic instructions are instructions for a RAM, and an 
execution environment consists of only one component.

After the presentation in question, we make precise in the setting of 
the presented theory what it means that a given instruction sequence 
computes a given partial function on bit strings, show that a relatively 
unknown, but more or less realistic, version of the RAM model of 
computation can be dealt with in this setting by imposing apposite 
restrictions, and introduce for this model an alternative to the usual 
time measures for versions of the RAM model.
Under the alternative time measure, RAMs from the version of the RAM 
model concerned can be simulated by multi-tape Turing machines with 
quadratic time overhead.
Moreover, under a usual space measure for versions of the RAM model, 
RAMs from this version of the RAM model can be simulated by multi-tape 
Turing machines with constant-factor space overhead.

With the instruction set of the version of the RAM model of computation
dealt with in this paper, a fairly realistic idealization of a real 
computer is obtained.
The introduced alternative to the usual time measures for versions of 
the RAM model has its origin in the simple idea that the time that it 
takes to execute an instruction from this instruction set should be 
based on the number of steps that a multi-tape Turing machine with input 
alphabet $\Bit$ needs to simulate the instruction.
It is to be expected that its instruction set makes the version of the 
RAM model dealt with in this paper very practical to the expression and
analysis of many algorithms.

This paper is organized as follows.
First, a survey is given of program algebra, basic thread algebra, and 
an extension of their combination that makes precise which behaviours 
are produced by instruction sequences under execution 
(Sections~\ref{sect-PGA}, \ref{sect-BTA}, and~\ref{sect-TE-BC}).
Next, the surveyed theory is instantiated and extended to handle 
interaction between instruction sequences (with instructions for a RAM) 
under execution and the memory of a RAM 
(Sections~\ref{sect-TTFA} and~\ref{sect-TSI}).
Then, in the setting of the resulting theory, it is made precise what it 
means that a given instruction sequence computes a given partial 
function (Section~\ref{sect-comput-boolstring-fnc}) and a more or less 
realistic version of the RAM model of computation is described
(Sections~\ref{sect-RAM-instructions}, \ref{sect-RAM-realistic}, 
and~\ref{sect-SRRAM-programs}).
After that, a new time measure and a known space measure for this 
model of computation are introduced (Sections~\ref{sect-bit-oriented}  
and~\ref{sect-bit-oriented-space}) and the former measure is discussed
(Section~\ref{sect-discussion}).
Finally, some concluding remarks are made (Section~\ref{sect-concl}).

In this paper, some familiarity with algebraic specification, 
computability, and computational complexity is assumed.
The relevant notions are explained in many handbook chapters and 
textbooks, e.g.~\cite{EM85a,ST12a,Wir90a} for the relevant notions 
concerning algebraic specification and~\cite{AHU74a,Mor98a,Pap94a} for 
the relevant notions concerning computability and computational 
complexity.

Sections~\ref{sect-PGA}--\ref{sect-TE-BC}, i.e.\ the preliminary 
sections of this paper, are largely shortened versions of Sections~2--4 
of~\cite{BM18b}, which, in turn, draw from the preliminary sections of 
several earlier papers.

\section{Program Algebra}
\label{sect-PGA}

This section presents a survey of program algebra (\PGA).
A program is perceived in \PGA\ as a single-pass instruction sequence, 
i.e.\ a possibly infinite sequence of instructions of which each 
instruction is executed at most once and can be dropped after it has 
been executed or jumped over.

It is assumed that a fixed but arbitrary set $\BInstr$ of 
\emph{basic instructions} has been given.
$\BInstr$ is the basis for the set of instructions that may occur in 
instruction sequences.
The intuition is that the execution of a basic instruction may modify a
state and must produce the value $\zero$ or $\one$ as reply at its 
completion.
The produced reply may be state-dependent.
In applications of \PGA, the instructions taken as basic instructions 
vary from instructions relating to Boolean registers to machine language 
instructions of actual computers.   

The set of instructions of which the instruction sequences are composed 
is the set that consists of the following elements:
\begin{itemize}
\item
for each $a \in \BInstr$, a \emph{plain basic instruction} $a$;
\item
for each $a \in \BInstr$, a \emph{positive test instruction} $\ptst{a}$;
\item
for each $a \in \BInstr$, a \emph{negative test instruction} $\ntst{a}$;
\item
for each $l \in \Nat$, a \emph{forward jump instruction} $\fjmp{l}$;
\item
a \emph{termination instruction} $\halt$.
\end{itemize}
We write $\PInstr$ for this set.
The elements from this set are called \emph{primitive instructions}.

On execution of an instruction sequence, the primitive instructions of 
which it is composed have the following effects:
\begin{itemize}
\item
the effect of a positive test instruction $\ptst{a}$ is that basic
instruction $a$ is executed and execution proceeds with the next
primitive instruction if the produced reply is $\one$ and otherwise the 
next primitive instruction is skipped and execution proceeds with the
primitive instruction following the skipped one --- inaction occurs if 
there is no primitive instruction to proceed with;
\item
the effect of a negative test instruction $\ntst{a}$ is the same as
the effect of $\ptst{a}$, but with the role of the produced reply
reversed;
\item
the effect of a plain basic instruction $a$ is the same as the effect
of $\ptst{a}$, but execution always proceeds as if the produced reply is
$\one$;
\item
the effect of a forward jump instruction $\fjmp{l}$ is that execution
proceeds with the $l$th next primitive instruction --- inaction occurs 
if $l$ equals $\zero$ or there is no primitive instruction to proceed with;
\item
the effect of the termination instruction $\halt$ is that execution 
terminates.
\end{itemize}
The phrase ``inaction occurs'' indicates that no more basic instructions 
are executed, but execution does not terminate.

\PGA\ has one sort: the sort $\InSeq$ of \emph{instruction sequences}. 
To build terms of sort $\InSeq$, \PGA\ has the following constants and 
operators:
\begin{itemize}
\item
for each $u \in \PInstr$, 
the \emph{instruction} constant $\const{u}{\InSeq}$\,;
\item
the binary \emph{concatenation} operator 
$\funct{\ph \conc \ph}{\InSeq \x \InSeq}{\InSeq}$\,;
\item
the unary \emph{repetition} operator 
$\funct{\ph\rep}{\InSeq}{\InSeq}$\,.
\end{itemize}
Terms of sort $\InSeq$ are built as usual in the one-sorted case.
We assume that there are infinitely many variables of sort $\InSeq$, 
including $X,Y,Z$.
We use infix notation for concatenation and postfix notation for
repetition.

A \PGA\ term in which the repetition operator does not occur is called 
a \emph{repetition-free} \PGA\ term.

One way of thinking about closed \PGA\ terms is that they represent 
non-empty, possibly infinite sequences of primitive instructions with 
finitely many distinct suffixes.
Let $t$ and $t'$ be closed \PGA\ terms representing instruction 
sequences $s$ and $s'$.
Then the operators of \PGA\ can be explained as follows:
\begin{itemize}
\item
$t \conc t'$ represents the concatenation of $s$ and $s'$;
\item 
$t\rep$ represents $s$ concatenated infinitely many times with itself.
\end{itemize}
 
The axioms of \PGA\ are given in Table~\ref{axioms-PGA}.%
\begin{table}[!t]
\caption{Axioms of \PGA} 
\label{axioms-PGA}
\begin{eqntbl}
\begin{axcol}
(X \conc Y) \conc Z = X \conc (Y \conc Z)             & \axiom{PGA1}  \\
(X^n)\rep = X\rep                                     & \axiom{PGA2}  \\
X\rep \conc Y = X\rep                                 & \axiom{PGA3}  \\
(X \conc Y)\rep = X \conc (Y \conc X)\rep             & \axiom{PGA4} 
\eqnsep
\fjmp{k{+}1} \conc u_1 \conc \ldots \conc u_k \conc \fjmp{0} =
\fjmp{0} \conc u_1 \conc \ldots \conc u_k \conc \fjmp{0} 
                                                      & \axiom{PGA5}  \\
\fjmp{k{+}1} \conc u_1 \conc \ldots \conc u_k \conc \fjmp{l} =
\fjmp{l{+}k{+}1} \conc u_1 \conc \ldots \conc u_k \conc \fjmp{l}
                                                      & \axiom{PGA6}  \\
(\fjmp{l{+}k{+}1} \conc u_1 \conc \ldots \conc u_k)\rep =
(\fjmp{l} \conc u_1 \conc \ldots \conc u_k)\rep       & \axiom{PGA7}  \\
\fjmp{l{+}k{+}k'{+}2} \conc u_1 \conc \ldots \conc u_k \conc
(v_1 \conc \ldots \conc v_{k'{+}1})\rep = {} \\ \phantom{{}{+}k'}
\fjmp{l{+}k{+}1} \conc u_1 \conc \ldots \conc u_k \conc
(v_1 \conc \ldots \conc v_{k'{+}1})\rep               & \axiom{PGA8} 
\end{axcol}
\end{eqntbl}
\end{table}
In this table, 
$u$, $u_1,\ldots,u_k$ and $v_1,\ldots,v_{k'+1}$ stand for arbitrary 
primitive instructions from $\PInstr$, 
$k$, $k'$, and $l$ stand for arbitrary natural numbers from $\Nat$, and
$n$ stands for an arbitrary natural number from $\Natpos$.%
\footnote
{We write $\Natpos$ for the set $\set{n \in \Nat \where n \geq 1}$ of
positive natural numbers.}
For each $n \in \Natpos$, the term $t^n$, where $t$ is a \PGA\ term, 
is defined by induction on $n$ as follows: $t^1 = t$, and 
$t^{n+1} = t \conc t^n$.

Let $t$ and $t'$ be closed \PGA\ terms.
Then $t = t'$ is derivable from the axioms of \PGA\ iff $t$ and $t'$ 
represent the same instruction sequence after changing all chained jumps 
into single jumps (which corresponds to applying axioms PGA5 and PGA6) 
and making all jumps as short as possible (which corresponds to applying 
axioms PGA7 and PGA8). 
Moreover, $t = t'$ is derivable from PGA1--PGA4 iff $t$ and $t'$ 
represent the same instruction sequence. 

The informal explanation of closed \PGA\ terms as sequences of primitive 
instructions given above can be looked upon as a sketch of the intended 
model of axioms PGA1--PGA4.
This model, which is described in detail in, for example, \cite{BM12b}, 
is an initial model of axioms PGA1--PGA4.
Henceforth, the instruction sequences of the kind considered in \PGA\ 
are called \PGA\ instruction sequences.

\section{Basic Thread Algebra for Finite and Infinite Threads}
\label{sect-BTA}

In this section, we introduce basic thread algebra (\BTA) and an 
extension of \BTA\ that reflects the idea that infinite threads are 
identical if their approximations up to any finite depth are identical.

\BTA\ is concerned with mathematical objects that model in a direct 
way the behaviours produced by \PGA\ instruction sequences under 
execution.
The objects in question are called threads.
A thread models a behaviour that consists of performing basic actions in 
a sequential fashion.
Upon performing a basic action, a reply from an execution environment
determines how the behaviour proceeds subsequently.
The possible replies are the values $\zero$ and~$\one$.

The basic instructions from $\BInstr$ are taken as basic actions.
Besides, $\Tau$ is taken as a special basic action.
It is assumed that $\Tau \notin \BAct$.
We write $\BActTau$ for $\BAct \union \set{\Tau}$.

\BTA\ has one sort: the sort $\Thr$ of \emph{threads}. 
To build terms of sort $\Thr$, \BTA\ has the following constants and 
operators:
\begin{itemize}
\item
the \emph{inaction} constant $\const{\DeadEnd}{\Thr}$;
\item
the \emph{termination} constant $\const{\Stop}{\Thr}$;
\item
for each $\alpha \in \BActTau$, the binary 
\emph{postconditional composition} operator 
$\funct{\pcc{\ph}{\alpha}{\ph}}{\linebreak[2]\Thr \x \Thr}{\Thr}$.
\end{itemize}
Terms of sort $\Thr$ are built as usual in the one-sorted case. 
We assume that there are infinitely many variables of sort $\Thr$, 
including $x,y,z$.
We use infix notation for postconditional composition.
We introduce \emph{basic action prefixing} as an abbreviation: 
$\alpha \bapf t$, where $\alpha \in \BActTau$ and $t$ is a \BTA\ term, 
abbreviates $\pcc{t}{\alpha}{t}$.
We treat an expression of the form $\alpha \bapf t$ and the \BTA\ term 
that it abbreviates as syntactically the same.

Closed \BTA\ terms are considered to represent threads.
The constants of \BTA\ can be explained as follows:
\begin{itemize}
\item
$\DeadEnd$ represents the thread that models inactive behaviour, i.e.\ 
the behaviour that performs no more basic actions and does not terminate
either; 
\item
$\Stop$ represents the thread that models the behaviour that does 
nothing else but terminate. 
\end{itemize}
Let $t$ and $t'$ be closed \BTA\ terms representing threads $r$ and 
$r'$.
Then the operators of \PGA\ can be explained as follows:
\begin{itemize}
\item
$\pcc{t}{\alpha}{t'}$ represents the thread that models the behaviour 
that first performs $\alpha$ and then proceeds as the behaviour modeled 
by $r$ if the reply from the execution environment is $\one$ and 
otherwise proceeds as the behaviour modeled by $r'$. 
\end{itemize}

\BTA\ has only one axiom.
This axiom is given in Table~\ref{axioms-BTA}.
\begin{table}[!t]
\caption{Axiom of \BTA} 
\label{axioms-BTA}
\begin{eqntbl}
\begin{axcol}
\pcc{x}{\Tau}{y} = \Tau \bapf x                             & \axiom{T1}
\end{axcol}
\end{eqntbl}
\end{table}
It tells us that performing $\Tau$, which is considered performing an 
internal action, always leads to the reply $\one$.

Each closed \BTA\ term represents a finite thread, i.e.\ a thread with 
a finite upper bound to the number of basic actions that it can perform.
Infinite threads, i.e.\ threads without such a finite upper bound, can 
be defined by means of a set of \emph{recursion equations}, i.e.\ a set 
$\set{x_i = t_i \where i \in I}$,
 where $I$ is an index set,
 each $x_i$ is a variable of sort $\Thr$, 
 each $t_i$ is a \BTA\ term in which only variables from
 $\set{x_i \where i \in I}$ occur, and  
 $x_i \neq x_j$ for all $i,j \in I$ with $i \neq j$.
A regular thread is a finite or infinite thread that can be defined by 
means of a finite set of recursion equations.
The behaviours produced by \PGA\ instruction sequences under execution 
are exactly the behaviours modeled by regular threads.

Two infinite threads are considered identical if their approximations up 
to any finite depth are identical.
The approximation up to depth $n$ of a thread models the behaviour that 
differs from the behaviour modeled by the thread in that it will become
inactive after it has performed $n$ actions unless it would terminate at
this point.
The approximation induction principle (AIP) is a conditional equation 
that formalizes the above-mentioned view on infinite threads.
In AIP, the approximation up to depth $n$ is phrased in terms of the
unary \emph{projection} operator $\funct{\proj{n}}{\Thr}{\Thr}$.

The axioms for the projection operators and AIP are given in
Table~\ref{axioms-BTAinf}.
\begin{table}[!t]
\caption{Axioms for the projection operators and AIP} 
\label{axioms-BTAinf}
\begin{eqntbl}
\begin{axcol}
\proj{0}(x) = \DeadEnd                                  & \axiom{PR1} \\
\proj{n+1}(\DeadEnd) = \DeadEnd                         & \axiom{PR2} \\
\proj{n+1}(\Stop) = \Stop                               & \axiom{PR3} \\
\proj{n+1}(\pcc{x}{\alpha}{y}) =
\pcc{\proj{n}(x)}{\alpha}{\proj{n}(y)}                  & \axiom{PR4}
\eqnsep
\LAND{n \geq 0} \proj{n}(x) = \proj{n}(y) \Limpl x = y  & \axiom{AIP}
\end{axcol}
\end{eqntbl}
\end{table}
In this table, $\alpha$ stands for an arbitrary basic action from 
$\BActTau$ and $n$ stands for an arbitrary natural number from $\Nat$.
We write \BTAinf\ for \BTA\ extended with the projection operators, 
the axioms for the projection operators, and~AIP.

Because we have to deal with conditional equational formulas with a 
countably infinite number of premises in \BTAinf, it is understood that
infinitary conditional equational logic is used in deriving equations 
from the axioms of \BTAinf.
A complete inference system for infinitary conditional equational logic 
can be found in~\cite{BW90,GV93,Gog21a}.

The depth of a finite thread is the maximum number of basic actions that 
it can perform before it terminates or becomes inactive.
We define the function $\depth$ that assigns to each closed \BTA\ term
the depth of the finite thread that it represents recursively as follows:
\begin{ldispl}
\depth(\Stop) = 0\;, \\
\depth(\DeadEnd) = 0\;, \\
\depth(\pcc{t}{\alpha}{t'}) = \max \set{\depth(t),\depth(t')} + 1\;.
\end{ldispl}%

\section{Thread Extraction from Instruction Sequences}
\label{sect-TE-BC}

In this section, we make precise in the setting of \BTAinf\ which 
behaviours are produced by \PGA\ instruction sequences under 
execution.

To make precise which behaviours are produced by \PGA\ instruction 
sequences under execution, we introduce an operator $\extr{\ph}$.
For each closed \PGA\ term $t$, $\extr{t}$ represents the thread that
models the behaviour produced by the instruction sequence represented 
by $t$ under execution.

Formally, we combine \PGA\ with \BTAinf\ and extend the combination with 
the \emph{thread extraction} operator $\funct{\extr{\ph}}{\InSeq}{\Thr}$ 
and the axioms given in Table~\ref{axioms-thread-extr}.%
\begin{table}[!t]
\caption{Axioms for the thread extraction operator} 
\label{axioms-thread-extr}
\begin{eqntbl}
\begin{axcol}
\extr{a} = a \bapf \DeadEnd                            & \axiom{TE1}  \\
\extr{a \conc X} = a \bapf \extr{X}                    & \axiom{TE2}  \\
\extr{\ptst{a}} = a \bapf \DeadEnd                     & \axiom{TE3}  \\
\extr{\ptst{a} \conc X} = \pcc{\extr{X}}{a}{\extr{\fjmp{2} \conc X}}
                                                       & \axiom{TE4}  \\
\extr{\ntst{a}} = a \bapf \DeadEnd                     & \axiom{TE5}  \\
\extr{\ntst{a} \conc X} = \pcc{\extr{\fjmp{2} \conc X}}{a}{\extr{X}}
                                                       & \axiom{TE6}
\end{axcol}
\qquad
\begin{axcol}
\extr{\fjmp{l}} = \DeadEnd                             & \axiom{TE7}  \\
\extr{\fjmp{0} \conc X} = \DeadEnd                     & \axiom{TE8}  \\
\extr{\fjmp{1} \conc X} = \extr{X}                     & \axiom{TE9}  \\
\extr{\fjmp{l+2} \conc u} = \DeadEnd                   & \axiom{TE10} \\
\extr{\fjmp{l+2} \conc u \conc X} = \extr{\fjmp{l+1} \conc X}
                                                       & \axiom{TE11} \\
\extr{\halt} = \Stop                                   & \axiom{TE12} \\
\extr{\halt \conc X} = \Stop                           & \axiom{TE13}
\end{axcol}
\end{eqntbl}
\end{table}
In this table, 
$a$ stands for an arbitrary basic instruction from $\BInstr$, 
$u$ stands for an arbitrary primitive instruction from $\PInstr$, and 
$l$ stands for an arbitrary natural number from $\Nat$.
We write \PGABTA\ for the combination of \PGA\ and \BTAinf\ extended 
with the thread extraction operator and the axioms for the thread 
extraction operator.

As mentioned in Section~\ref{sect-PGA}, on execution of an instruction 
sequence, inaction occurs if there is no primitive instruction to 
proceed with.
That is why $\DeadEnd$ occurs in axioms TE1, TE3, TE5, TE7, and TE10.

If a closed \PGA\ term $t$ represents an infinite instruction sequence, 
then we can extract the approximations of the thread modeling the 
behaviour produced by that instruction sequence under execution up to 
every finite depth: for each $n \in \Nat$, there exists a closed \BTA\ 
term $t''$ such that $\proj{n}(\extr{t}) = t''$ is derivable from axioms 
PGA1--PGA8, PR1--PR4, and TE1--TE13.
If closed \PGA\ terms $t$ and $t'$ represent infinite instruction 
sequences that produce the same behaviour under execution, then this can
be proved using the following instance of AIP:
$\LAND{n \geq 0} \proj{n}(\extr{t}) = \proj{n}(\extr{t'}) \Limpl
 \extr{t} = \extr{t'}$.

If a closed \PGA\ term $t$ represents an instruction sequence that
starts with an infinite chain of forward jumps, then TE9 and TE11 can 
be applied to $\extr{t}$ infinitely often without ever showing that a 
basic action is performed.
In this case, we have to do with inaction and, being consistent with 
that, $\extr{t} = \DeadEnd$ is derivable from axioms PGA1--PGA8, 
PR1--PR4, AIP, and TE1--TE13.
By contrast, $\extr{t} = \DeadEnd$ is not derivable from axioms 
PGA1--PGA4, PR1--PR4, AIP, and TE1--TE13.
However, if closed \PGA\ terms $t$ and $t'$ represent instruction 
sequences in which no infinite chains of forward jumps occur, then 
$t = t'$ is derivable from the axioms of \PGA\ only if 
$\extr{t} = \extr{t'}$ is derivable from PGA1--PGA4, PR1--PR4, AIP, and 
TE1--TE13.

The following proposition, proved in~\cite{BM12b}, puts the 
expressiveness of \PGA\ in terms of producible behaviours.
\begin{proposition}
\label{proposition-expr}
Let $\mathcal{M}$ be a model of \PGABTA.
Then, for each element $r$ from the domain associated with the sort 
$\Thr$ in $\mathcal{M}$, there exists a closed \PGA\ term $t$ such that 
the interpretation of $\extr{t}$ in $\mathcal{M}$ is $r$ iff 
$r$ can be defined by means of a finite set of recursion equations.
\end{proposition}

\PGA\ instruction sequences are behaviourally equivalent if they 
produce the same behaviour under execution.
Instruction sequences are behaviourally congruent if they produce the
same behaviour irrespective of the way they are entered and the way
they are left during execution (see also~\cite{BL02a,BM12b}).

\section{Basic Instructions for Random Access Machines}
\label{sect-TTFA}

\PGA\ instruction sequences under execution may interact with components 
of their execution environment.
The execution environment components vary from one application of \PGA\
to another.
In this section, we consider basic instruction for the case where the 
execution environment components are memories of RAMs.

The memory of a RAM consists of a countably infinite number of registers 
which are numbered by natural numbers.
Each register is capable of containing a bit string of arbitrary length.
The contents of the registers constitute the state of the memory.

A \emph{RAM memory state} is a function 
$\funct{\sigma}{\Nat}{\BitStr}$ 
that satisfies the condition that there exists a $i \in \Nat$ such that, 
for all $j \in \Nat$, $\sigma(i + j) = \emptystr$.%
\footnote{We write $\emptystr$ for the empty bit string.} 
We write $\MStates$ for the set of all RAM memory states, and we write
$\sigma_\emptystr$ for the unique $\sigma \in \MStates$ such that 
$\sigma(i) = \emptystr$ for all $i \in \Nat$.

Let $\sigma$ be a RAM memory state.
Then, for all $i \in \Nat$, $\sigma(i)$ is the content of the register 
with number $i$ in memory state $\sigma$.
The condition expresses that the part of the memory that is actually in 
use remains finite.

Henceforth, we will use the notation $\fncvn{\sigma}{i_1}{w_1}{i_n}{w_n}$.
For each $\funct{\sigma}{\Nat}{\BitStr}$, $i_1,\ldots,i_n \in \Nat$, and 
$w_1,\ldots,w_n \in \BitStr$, $\fncvn{\sigma}{i_1}{w_1}{i_n}{w_n}$ is the 
function $\funct{\sigma'}{\Nat}{\BitStr}$ defined as follows: 
$\sigma'(i_1) = w_1$, \ldots, $\sigma'(i_n) = w_n$, and, for all 
$j \in \Nat$ with $j \notin \set{i_1,\ldots,i_n}$, 
$\sigma'(j) = \sigma(j)$.

The execution of an instruction by a RAM may change the memory state of 
the RAM and must produce the value $\zero$ or $\one$ as reply.

The set of basic instructions used in this case consists of a 
\emph{basic RAM instruction} $\iram{p}{q}$ 
for each $\funct{p}{\MStates}{\Bit}$ and
$\funct{q}{\MStates}{\MStates}$ 
that satisfy the following conditions (which are explained below) for 
all $\sigma \in \MStates$:
\begin{itemize}
\item[(a)]
there exists at most one $i \in \Nat$ for which 
there exists a $w \in \BitStr$ such that 
$\sigma(i) \neq w$ and $p(\sigma) \neq p(\fncv{\sigma}{i}{w})$,
\item[(b)]
there exists at most one $i \in \Nat$ for which 
$\sigma(i) \neq q(\sigma)(i)$,
\item[(c)]
if
there exists an $i \in \Nat$ for which 
there exists a $w \in \BitStr$ such that 
$\sigma(i) \neq w$ and $p(\sigma) \neq p(\fncv{\sigma}{i}{w})$ and
there exists an $i \in \Nat$ for which $\sigma(i) \neq q(\sigma)(i)$,
then 
there exists an $i \in \Nat$ for which 
there exists a $w \in \BitStr$ such that 
$\sigma(i) \neq w$ and $p(\sigma) \neq p(\fncv{\sigma}{i}{w})$ and
$\sigma(i) \neq q(\sigma)(i)$.
\end{itemize}
We write $\BInstrram$ for this set.

Each basic RAM instruction leads to carrying out an operation on a RAM 
memory when the instruction is executed.
The intuition is basically that carrying out the operation modifies the 
content of a single register of the RAM memory and produces the value
$\zero$ or $\one$ as reply depending on the content of this register.
More precisely, the execution of a basic RAM instruction $\iram{p}{q}$ 
has the following effects:
\begin{itemize}
\item
if the RAM memory state is $\sigma$ when the execution of $\iram{p}{q}$ 
starts, then the reply produced on termination of the execution of 
$\iram{p}{q}$ is $p(\sigma)$;
\item
if the RAM memory state is $\sigma$ when the execution of $\iram{p}{q}$ 
starts, then the RAM memory state is $q(\sigma)$ when the execution of 
$\iram{p}{q}$ terminates.
\end{itemize}
Condition~(a) expresses that a basic RAM instruction does not produce a 
reply that depends on the content of more than one register.
Condition~(b) expresses that a basic RAM instruction does not modify the 
content of more than one register.
Condition~(c) expresses that a basic RAM instruction produces a reply 
that depends on the content of a register and modifies the content of a 
register only if the former register is the same as the latter register.

A function from $\MStates$ to $\Bit$ for which condition~(a) 
trivially holds is the function $\OFunc$ defined by 
$\OFunc(\sigma) = \one$.
A function from $\MStates$ to $\MStates$ for which condition~(b) 
trivially holds is the function $\IFunc$ defined by 
$\IFunc(\sigma) = \sigma$.
From Section~\ref{sect-RAM-realistic}, only basic RAM instruction of the 
forms $\iram{\OFunc}{q}$ and $\iram{p}{\IFunc}$ are considered.

We write $\PGABTAram$ for \PGABTA\ with $\BInstr$ instantiated 
by~$\BInstrram$.

\section{Interaction of Threads with RAM Memories}
\label{sect-TSI}

If instructions from $\BInstrram$ are taken as basic instructions, a 
\PGA\ instruction sequence under execution may interact with the memory 
of a RAM.
In line with this kind of interaction, a thread may perform a basic 
action basically for the purpose of changing 
the memory state of a RAM or receiving a reply that depends on the 
memory state of a RAM.
In this section, we introduce related constants and operators.

We extend $\PGABTA(\BInstrram)$ with the sort $\RAM$ of 
\emph{RAM memories}, the following operators:
\begin{itemize}
\item
for each $\sigma \in \MStates \union \set{\Div}$, 
the \emph{RAM memory} constant 
$\const{\sram{\sigma}}{\RAM}$;
\item
the binary \emph{use} operator
$\funct{\ph \sfuse \ph}{\Thr \x \RAM}{\Thr}$;
\item
the binary \emph{apply} operator
$\funct{\ph \sfapply \ph}{\Thr \x \RAM}{\RAM}$;
\end{itemize}
and the axioms given in Tables~\ref{axioms-use-apply}.%
\footnote
{We write $t[t'/x]$ for the result of substituting term $t'$ for 
variable $x$ in term $t$.}
\begin{table}[!t]
\caption{Axioms for the use and apply operator} 
\label{axioms-use-apply}
\begin{eqntbl}
\begin{saxcol}
\Stop  \sfuse u = \Stop                                & & \axiom{U1} \\
\DeadEnd \sfuse u = \DeadEnd                           & & \axiom{U2} \\
(\Tau \bapf x) \sfuse u = \Tau \bapf (x \sfuse u)      & & \axiom{U3} \\
(\pcc{x}{\iram{p}{q}}{y}) \sfuse \sram{\sigma} = 
\Tau \bapf (x \sfuse \sram{q(\sigma)})
                                & \mif p(\sigma) = \one  & \axiom{U4} \\
(\pcc{x}{\iram{p}{q}}{y}) \sfuse \sram{\sigma} =
\Tau \bapf (y \sfuse \sram{q(\sigma)}) 
                                & \mif p(\sigma) = \zero & \axiom{U5} \\
(\pcc{x}{\iram{p}{q}}{y}) \sfuse \sram{\Div} = 
\Tau \bapf \DeadEnd                                    & & \axiom{U6} \\
\proj{n}(x \sfuse u) = \proj{n}(x) \sfuse u            & & \axiom{U7} 
\eqnsep
\Stop  \sfapply u = u                                  & & \axiom{A1} \\
\DeadEnd \sfapply u = \sram{\Div}                      & & \axiom{A2} \\
(\Tau \bapf x) \sfapply u = \Tau \bapf (x \sfapply u)  & & \axiom{A3} \\
(\pcc{x}{\iram{p}{q}}{y}) \sfapply \sram{\sigma} = 
x \sfapply \sram{q(\sigma)}     & \mif p(\sigma) = \one  & \axiom{A4} \\
(\pcc{x}{\iram{p}{q}}{y}) \sfapply \sram{\sigma} = 
y \sfapply \sram{q(\sigma)}     & \mif p(\sigma) = \zero & \axiom{A5} \\
(\pcc{x}{\iram{p}{q}}{y}) \sfapply \sram{\Div} = 
\sram{\Div}                                            & & \axiom{A6} \\
\LAND{k \geq n} t[\proj{k}(x)/z] = t'[\proj{k}(y)/z] \Limpl 
t[x/z] = t'[y/z]                                       & & \axiom{A7}
\end{saxcol}
\end{eqntbl}
\end{table}
In these tables, 
$p$ stands for an arbitrary function from $\MStates$ to $\Bit$,
$q$ stands for an arbitrary function from $\MStates$ to $\MStates$, 
$\sigma$ stands for an arbitrary RAM memory state from $\MStates$, 
$n$ stands for an arbitrary natural number from $\Nat$, and
$t$ and $t'$ stand for arbitrary terms of sort $\RAM$.
Moreover, $u$ is assumed to be a variable of sort $\RAM$.
We use infix notation for the use and apply operators.
We write \PGABTARAMI\ for \PGABTAram\ extended with the sort $\RAM$,
the RAM memory constants, the use operator, the apply operator, and the 
axioms for these operators.

Axioms U1--U6 and A1--A6 formalize the informal explanation of the use 
operator and the apply operator given below and in addition stipulate 
what is the result of use and apply if an inoperative RAM memory is 
involved (U6 and A6).
Axioms U7 and A7 allow for reasoning about infinite threads, and 
therefore about the behaviour produced by infinite instruction sequences 
under execution, in the context of use and apply, respectively.

The RAM memory denoted by a closed term of the form $\sram{\sigma}$, 
where $\sigma \in \MStates$, is an operative RAM memory whose state is 
$\sigma$.
The RAM memory denoted by a closed term of the form $\sram{\Div}$ is an 
inoperative RAM memory.
An inoperative RAM memory can be viewed as a RAM memory whose state is 
unavailable.
Carrying out an operation on an inoperative RAM memory is impossible.

On interaction between a thread and a RAM memory, the thread affects the 
RAM memory and the RAM memory affects the thread.
The use operator concerns the effects of a RAM memory on a thread and 
the apply operator concerns the effects of a thread on a RAM memory.
The thread denoted by a closed term of the form $t \sfuse t'$ and the
RAM memory denoted by a closed term of the form $t \sfapply t'$ are the 
thread and RAM memory, respectively, that result from carrying out the 
operations that go with the basic actions performed by the thread 
denoted by $t$ on the RAM memory denoted by $t'$.
When the operation that goes with a basic action performed by a thread 
is carried out on a RAM memory, the state of the RAM memory is changed 
according to the operation concerned and the thread is affected as 
follows: the basic action turns into the internal action $\Tau$ and 
the two ways to proceed reduce to one on the basis of the reply produced 
according to the operation concerned.
Thus, the internal action $\Tau$ is left as a trace of each basic action 
that has led to carrying out an operation on the RAM memory.

The following two elimination results for closed \PGABTARAMI\ terms are 
proved similarly to Theorems~1 and~2 from~\cite{BM19a}.
\sloppy
\begin{proposition}
\label{proposition-elim-use}
For all closed \PGABTARAMI\ terms $t$ of sort $\,\Thr\,$ in 
\linebreak[2] 
which all subterms of sort $\InSeq$ are repetition-free, there exists a 
closed \PGABTAram\ term $t'$ of sort $\Thr$ such that $t = t'$ is 
derivable from the axioms of \PGABTARAMI.
\end{proposition}
\begin{proposition}
\label{proposition-elim-apply}
For all closed \PGABTARAMI\ terms $t$ of sort $\RAM$ in which all 
subterms of sort $\InSeq$ are repetition-free, there exists a closed 
\PGABTAram\ term $t'$ of sort $\RAM$ such that $t = t'$ is derivable 
from the axioms of \PGABTARAMI.
\end{proposition}

\section{Computing Partial Functions from \boldmath{${(\BitStr)}^n$} to 
  \boldmath{$\BitStr$}}
\label{sect-comput-boolstring-fnc}

In this section, we make precise in the setting of the algebraic theory
\PGABTARAMI\ what it means that a given instruction sequence computes a 
given partial function from ${(\BitStr)}^n$ to $\BitStr$ ($n \in \Nat$).

We use the notation $\pfunct{f}{A}{B}$ to indicate that $f$ is a partial 
function from $A$ to $B$.
We write $\len(w)$, where $w \in \Bit^*$, for the length of $w$.

Let $t$ be a closed \PGABTARAMI\ term of sort $\InSeq$,
let $n \in \Nat$, let $\pfunct{F}{{(\BitStr)}^n}{\BitStr}$, and
let $\funct{T}{\Nat}{\Nat}$.
Then $t$ \emph{computes $F$ in time~$T$ under the uniform time measure} 
if:
\begin{itemize}
\item
for all $w_1,\ldots,w_n \in \BitStr$ such that $F(w_1,\ldots,w_n)$ 
is defined, there exists a $\sigma \in \MStates$ such that: 
\begin{ldispl}
\extr{t} \sfapply \sram{\fncvn{\sigma_\emptystr}{1}{w_1}{n}{w_n}} = 
\sram{\fncv{\sigma}{0}{F(w_1,\ldots,w_n)}}\;,
\seqnsep 
\depth(\extr{t} \sfuse \sram{\fncvn{\sigma_\emptystr}{1}{w_1}{n}{w_n}})
 \leq
T(\len(w_1) + \ldots + \len(w_n))\;; 
\end{ldispl}%
\item
for all $w_1,\ldots,w_n \in \BitStr$ such that $F(w_1,\ldots,w_n)$ is
undefined:
\begin{ldispl}
\extr{t} \sfapply \sram{\fncvn{\sigma_\emptystr}{1}{w_1}{n}{w_n}} = 
\sram{\Div}\;.
\end{ldispl}%
\end{itemize}
We say that $t$ \emph{computes $F$} if there exists a 
$\funct{T}{\Nat}{\Nat}$ such that $t$ computes $F$ in time $T$ under the 
uniform time measure.

With the above definition, we can establish whether an instruction 
sequence of the kind considered in \PGABTARAMI\ computes a given partial 
function from ${(\BitStr)}^n$ to $\BitStr$ ($n \in \Nat$) mainly by 
equational reasoning using the axioms of \PGABTARAMI.
The axioms for the apply operator given in Table~\ref{axioms-use-apply}, 
i.e.\ axioms A1--A7, are instrumental in that.

The setting provided by \PGABTARAMI\ is more general than the setting 
provided by any known version of the RAM model of computation.
\PGABTARAMI\ is not suitable as a model of computation itself, but 
virtually all known versions of the RAM model of computation can be 
dealt with by imposing restrictions on the set of basic RAM instructions
($\BInstrram$).
Investigations of issues in areas such as computational complexity and 
analysis of algorithms require restriction to instructions that are 
found to be sufficiently primitive.
Without any restriction on $\BInstrram$, we even have that, for each 
computable {$\pfunct{F}{\nolinebreak{(\BitStr)}^n}{\BitStr}$}, 
there exists a closed \PGABTARAMI\ term $t$ of sort $\InSeq$ such that 
$t$ computes $F$ in one step. 

Restriction of the set of basic RAM instructions to instructions, with 
both direct and indirect addressing of registers, to carry out addition 
by one on natural numbers, to carry out comparisons of natural numbers 
on equal to and greater than, and to copy natural numbers (identifying 
bit strings with the natural numbers that they represent) gives rise to 
the version of the RAM model of computation known as the successor RAM 
model.
The basic instructions of a successor RAM are clearly very primitive, 
but as a consequence of that a successor RAM is not a very realistic 
idealization of a real computer.
In Section~\ref{sect-RAM-instructions}, we present a set of basic RAM 
instructions that yields a much more realistic idealization of a real 
computer.

Whatever version of the RAM model of computation is obtained by 
restriction of the set of basic RAM instructions considered in 
\PGABTARAMI, it is an idealization of a real computer in the sense that 
its memory offers an unbounded number of registers that can contain a 
bit string of arbitrary length instead of a bounded number of registers 
that can only contain a bit string of a fixed length.

\section{Basic Instructions for More or Less Realistic RAMs}
\label{sect-RAM-instructions}

In this section, we introduce a set of basic RAM instructions that give
rise to a version of the RAM model of computation that is a fairly 
realistic idealization of a real computer.

In general, the execution of an instruction by a real computer changes
the memory state of the computer by carrying out a certain operation on 
the contents of certain registers and changing the content of a certain 
register into the result of this.
We use a special notation reflecting this for the restricted set of 
basic RAM instructions with which a fairly realistic idealization of a 
real computer is obtained.
This restricted set of basic RAM instructions consists of all basic RAM 
instructions that have one of the following forms in the special 
notation:
\begin{ldispl}
\binop{:}\src_1{:}\src_2{:}\dst \quad \mathrm{or} \quad 
\unop{:}\src_1{:}\dst           \quad \mathrm{or} \quad 
\cmpop{:}\src_1{:}\src_2\;,
\end{ldispl}%
where 
\begin{ldispl}
\begin{array}[t]{@{}l@{\,\,}l@{\,\,}l@{}}
\binop & \in &
         \set{\addop,\subop,\mulop,\divop,\andop,\orop,\xorop}\;,
\\
\unop  & \in & \set{\notop,\shlop,\shrop,\rolop,\rorop,\movop}\;,
\\
\cmpop & \in & \set{\eqop,\gtop}\;.
\end{array}
\end{ldispl}%
\hsp{.6} and
\pagebreak[2]
\begin{ldispl}
\begin{tabular}[t]{@{}l@{\,\,}l@{\,\,}r@{}}
$\src_1$ & has one of the following forms: &
$\# i \;\, \mathrm{or} \;\, 
 i    \;\, \mathrm{or} \;\, 
 @ i, \;   \mathrm{where} \; i \in \Nat$, 
\\
$\src_2$ & has one of the following forms: &
$\# i \;\, \mathrm{or} \;\, 
 i    \;\, \mathrm{or} \;\, 
 @ i, \;   \mathrm{where} \; i \in \Nat$, 
\\
$\dst$   & has one of the following forms: &
$i \;\, \mathrm{or} \;\,
 @ i, \; \mathrm{where} \; i \in \Nat$, 
\end{tabular}
\end{ldispl}%
We write $\BInstrrram$ for this set of basic RAM instructions.
Moreover, we write \linebreak[2] $\Src$ for the set
$\set{\# i \where i \in \Nat} \union \Nat \union
 \set{@ i \where i \in \Nat}$, 
$\Dst$ for the set
$\Nat \union \set{@ i \where i \in \Nat}$,
and $\CInstrrram$ for the set
$\set{\cmpop{:}s_1{:}s_2 \where
      \cmpop \in \set{\eqop,\gtop} \Land s_1,s_2 \in \Src}$.

The following is a preliminary explanation of basic RAM instructions of 
the different forms:
\begin{itemize}
\item
on execution of an instruction of the form 
$\binop{:}\src_1{:}\src_2{:}\dst$, 
the binary operation named $\binop$ is carried out on the values given 
by $\src_1$ and $\src_2$ and the content of the register given by $\dst$ 
is changed into the result of this;
\item
on execution of an instruction of the form $\unop{:}\src_1{:}\dst$, 
the unary operation named $\unop$ is carried out on the value given by 
$\src_1$ and the content of the register given by $\dst$ is changed into
the result of this;
\item
on execution of an instruction of the form $\cmpop{:}\src_1{:}\src_2$, 
the comparison operation named $\cmpop$ is carried out on the values 
given by $\src_1$ and $\src_2$ and the result of this is produced as 
reply.
\end{itemize}

For each of the basic RAM instructions from $\BInstrrram$, each operand 
of the operation to be carried out on its execution is given in one the 
following three ways:
\begin{itemize}
\item
\emph{immediate}: it is the shortest bit string representing the natural 
number $i$ if $\src$ is of the form $\# i$;
\item
\emph{direct addressing}: it is the content of the register with number 
$i$ if $\src$ is of the form $i$;
\item
\emph{indirect addressing}: it is the content of the register whose 
number is represented by the content of the register with number $i$ if 
$\src$ is of the form~$@ i$.
\end{itemize}
Except for the comparison instructions, the result of the operation 
concerned becomes the content of a register in one the following two 
ways:
\begin{itemize}
\item
\emph{direct addressing}: it becomes the content of the register with 
number $i$ if $\dst$ is of the form $i$;
\item
\emph{indirect addressing}: it becomes the content of the register whose 
number is represented by the content of the register with number $i$ if 
$\dst$ is of the form~$@ i$.
\end{itemize}
As mentioned above, in the case of comparison instructions, the result 
of the operation concerned becomes the reply produced.

The following kinds of instructions are included in $\BInstrrram$:
\begin{itemize}
\item
\emph{arithmetic} instructions ($\addop,\subop,\mulop,\divop$) for 
carrying out operations that model arithmetic operations on natural 
numbers with respect to their binary representation by bit strings;
\item
\emph{logical} instructions ($\andop,\orop,\xorop,\notop$) for carrying 
out bitwise logical operations on bit strings;
\item
\emph{shift/rotate} instructions ($\shlop,\shrop,\rolop,\rorop$) for 
carrying out bit shift and rotate operations on bit strings;
\item
\emph{data transfer} instructions ($\movop$) for copying bit strings;
\item
\emph{comparison} instructions ($\eqop,\gtop$) for carrying out 
comparison operations on bit strings.
\end{itemize}
Data transfer instructions can be interpreted as instructions for 
carrying out the identity operation on bit strings.

Virtually all common general-purpose instructions of real computers are 
essential\-ly variants of the basic RAM instructions from $\BInstrrram$. 
Therefore, we believe that $\BInstrrram$ yields a version of the RAM 
model of computation that is a fairly realistic idealization of a real 
computer.

Above, a special notation is used for the basic RAM instructions from 
the set $\BInstrrram$.
In order to use the version of the RAM model of computation with this
set of basic RAM instructions in the setting of \PGABTARAMI, the special 
notation must be related to the notation used in that setting.

\section{More or Less Realistic RAM Instructions and \boldmath{\PGABTARAMI}}
\label{sect-RAM-realistic}

In this section, we relate the special notation for basic RAM 
instructions used in Section~\ref{sect-RAM-instructions} to the notation 
used in the setting of \PGABTARAMI.

We start with defining auxiliary functions for conversion between 
natural numbers and bit strings and evaluation of the elements of $\Src$ 
and $\Dst$. 

We write $\monus$ for proper subtraction of natural numbers.
We write $\div$ for zero-totalized Euclidean division of natural 
numbers, i.e.\ Euclidean division made total by imposing that division 
by zero yields zero (like in meadows, see \mbox{e.g.~\cite{BT07a,BM09g}}).
We use juxtaposition for concatenation of bit strings.

The \emph{natural to bit string} function 
$\funct{\ntob}{\Nat}{\BitStr}$ 
is recursively defined as follows:
\begin{itemize}
\item[] 
$\ntob(b) = b$ and 
$\ntob(n) = (n \bmod 2) \cat \ntob(n \div 2)$ if $n > 1$
\end{itemize}
and the \emph{bit string to natural} function 
$\funct{\bton}{\BitStr}{\Nat}$ 
is recursively defined as follows: 
\begin{itemize}
\item[] 
$\bton(\emptystr) = 0$ and 
$\bton(b \cat w) =\linebreak[2] 2 \mul \bton(w) + b$.
\end{itemize}
These definitions tell us that, when viewed as the binary representation 
of a natural number, the first bit of a bit string is considered the 
least significant bit.
Results of applying $\ntob$ have no leading zeros, but the operand of 
$\bton$ may have leading zeros.
Thus, we have that $\bton(\ntob(n)) = n$ and $\ntob(\bton(w)) = w'$, 
where $w'$ is $w$ without leading zeros.

For each $\sigma \in \MStates$, the \emph{src-valuation in $\sigma$} 
function $\funct{\val{\sigma}}{\Src}{\BitStr}$ is defined as follows:
\begin{itemize}
\item[] 
$\val{\sigma}(\# i) = \ntob(i)$, $\val{\sigma}(i) = \sigma(i)$, and
$\val{\sigma}(@ i) = \sigma(\bton(\sigma(i)))$  
\end{itemize}
and, 
for each $\sigma \in \MStates$, the \emph{dst-valuation in $\sigma$} 
function $\funct{\reg{\sigma}}{\Dst}{\Nat}$ is defined as follows:
\begin{itemize}
\item[] 
$\reg{\sigma}(i) = i$ and $\reg{\sigma}(@ i) = \bton(\sigma(i))$. 
\end{itemize}

We define the operations on bit strings that the operation names
$\addop$, $\subop$, $\mulop$, and $\divop$ refer to as follows:
\begin{ldispl}
\begin{tabular}[t]{@{}l@{\,\,}l@{}}
$\funct{+}{\BitStr \x \BitStr}{\BitStr}$: &
$w_1 + w_2 = \ntob(\bton(w_1) + \bton(w_2))$;
\\
$\funct{\monus}{\BitStr \x \BitStr}{\BitStr}$: &
$w_1 \monus w_2 = \ntob(\bton(w_1) \monus \bton(w_2))$;
\\
$\funct{\hsp{.25}\mul\hsp{.25}}{\BitStr \x \BitStr}{\BitStr}$: &
$w_1 \hsp{.25}\mul\hsp{.25} w_2 = 
\ntob(\bton(w_1) \hsp{.25}\mul\hsp{.25} \bton(w_2))$;
\\
$\funct{\div}{\BitStr \x \BitStr}{\BitStr}$: &
$w_1 \div w_2 = \ntob(\bton(w_1) \div \bton(w_2))$.
\end{tabular}
\end{ldispl}%
These definitions tell us that, although the operands of the operations 
$+$, $\monus$, $\mul$, and $\div$ may have leading zeros, results of 
applying these operations have no leading zeros.

We define the operations on bit strings that the operation names
$\andop$, $\orop$, $\xorop$, and $\notop$ refer to recursively as 
follows:
\begin{ldispl}
\begin{tabular}[t]{@{}l@{}}
$\funct{\Land}{\BitStr \x \BitStr}{\BitStr}$:\hsp{.25}
$\emptystr \Land \emptystr = \emptystr$,\hsp{.25}
$\emptystr \Land (b \cat w) = \zero \cat (\emptystr \Land w)$, 
\\ \hsp{.75}
$(b \cat w) \Land \emptystr = \zero \cat (w \Land \emptystr)$,\hsp{.25}
$(b_1 \cat w_1) \Land (b_2 \cat w_2) =
 (b_1 \Land b_2) \cat (w_1 \Land w_2)$;
\seqnsep
$\funct{\Lor}{\BitStr \x \BitStr}{\BitStr}$:\hsp{.25}
$\emptystr \Lor \emptystr = \emptystr$,\hsp{.25}
$\emptystr \Lor (b \cat w) = b \cat (\emptystr \Lor w)$, \\ \hsp{.75}
$(b \cat w) \Lor \emptystr = b \cat (w \Lor \emptystr)$,\hsp{.25}
$(b_1 \cat w_1) \Lor (b_2 \cat w_2) =
 (b_1 \Lor b_2) \cat (w_1 \Lor w_2)$;
\seqnsep
$\funct{\Lxor}{\BitStr \x \BitStr}{\BitStr}$:\hsp{.25}
$\emptystr \Lxor \emptystr = \emptystr$,\hsp{.25}
$\emptystr \Lxor (b \cat w) = b \cat (\emptystr \Lxor w)$, \\ \hsp{.75}
$(b \cat w) \Lxor \emptystr = b \cat (w \Lxor \emptystr)$,\hsp{.25}
$(b_1 \cat w_1) \Lxor (b_2 \cat w_2) =
 (b_1 \Lxor b_2) \cat (w_1 \Lxor w_2)$;
\seqnsep
$\funct{\Lnot}{\BitStr}{\BitStr}$:\hsp{.25}
$\Lnot \emptystr = \emptystr$,\hsp{.25}
$\Lnot (b \cat w) = (\Lnot b) \cat (\Lnot w)$.
\end{tabular}
\end{ldispl}%
These definitions tell us that, if the operands of the operations 
$\Land$, $\Lor$, and $\Lxor$ do not have the same length, sufficient 
leading zeros are assumed to exist.
Moreover, results of applying these operations and results of applying 
$\Lnot$ can have leading zeros.

We define the operations on bit strings that the operation names
$\shlop$, $\shrop$, $\rolop$, and $\rorop$ refer to as follows:
\begin{ldispl}
\begin{tabular}[t]{@{}l@{\,\,}l@{}}
$\funct{\shl}{\BitStr}{\BitStr}$: & 
$\shl \emptystr = \emptystr$,\hsp{.25} 
$\shl (b \cat w) = \zero \cat b \cat w$;
\\
$\funct{\shr}{\BitStr}{\BitStr}$: &
$\shr \emptystr = \emptystr$,\hsp{.25} 
$\shr (b \cat w) = w$;
\\
$\funct{\rol}{\BitStr}{\BitStr}$: &
$\rol \emptystr = \emptystr$,\hsp{.25} 
$\rol (w \cat b) = b \cat w$;
\\
$\funct{\ror}{\BitStr}{\BitStr}$: &
$\ror \emptystr = \emptystr$,\hsp{.25} 
$\ror (b \cat w) = w \cat b$.
\end{tabular}
\end{ldispl}%
These definitions tell us that results of applying the operations 
$\shl$, $\shr$, $\rol$, and $\ror$ can have leading zeros.
We have that $\bton(\shl w) = \bton(w) \mul 2$ and 
$\bton(\shr w) = \bton(w) \div 2$.

Now, we are ready to relate the special notation for basic RAM 
instructions used in Section~\ref{sect-RAM-instructions} to the notation 
used in the setting of \PGABTARAMI:
\begin{ldispl}
\begin{tabular}[t]{@{}l@{\,\,}l@{\,\,}l@{}}
$\addop{:}s_1{:}s_2{:}d$ & stands for $\iram{\OFunc}{q}$ where &
$q(\sigma) =
 \fncv{\sigma}{\reg{\sigma}(d)}{\val{\sigma}(s_1) + \val{\sigma}(s_2)}$;
\\
$\subop{:}s_1{:}s_2{:}d$ & stands for $\iram{\OFunc}{q}$ where &
$q(\sigma) =
 \fncv{\sigma}{\reg{\sigma}(d)}
  {\val{\sigma}(s_1) \monus \val{\sigma}(s_2)}$;
\\
$\mulop{:}s_1{:}s_2{:}d$ & stands for $\iram{\OFunc}{q}$ where &
$q(\sigma) =
 \fncv{\sigma}{\reg{\sigma}(d)}
  {\val{\sigma}(s_1) \mul \val{\sigma}(s_2)}$;
\\
$\divop{:}s_1{:}s_2{:}d$ & stands for $\iram{\OFunc}{q}$ where &
$q(\sigma) =
 \fncv{\sigma}{\reg{\sigma}(d)}
  {\val{\sigma}(s_1) \div \val{\sigma}(s_2)}$;
\\
$\andop{:}s_1{:}s_2{:}d$ & stands for $\iram{\OFunc}{q}$ where &
$q(\sigma) =
 \fncv{\sigma}{\reg{\sigma}(d)}
  {\val{\sigma}(s_1) \Land \val{\sigma}(s_2)}$;
\\
$\orop{:}s_1{:}s_2{:}d$ & stands for $\iram{\OFunc}{q}$ where &
$q(\sigma) =
 \fncv{\sigma}{\reg{\sigma}(d)}
  {\val{\sigma}(s_1) \Lor \val{\sigma}(s_2)}$;
\\
$\xorop{:}s_1{:}s_2{:}d$ & stands for $\iram{\OFunc}{q}$ where &
$q(\sigma) =
 \fncv{\sigma}{\reg{\sigma}(d)}
  {\val{\sigma}(s_1) \Lxor \val{\sigma}(s_2)}$;
\\
$\notop{:}s_1{:}d$ & stands for $\iram{\OFunc}{q}$ where &
$q(\sigma) =
 \fncv{\sigma}{\reg{\sigma}(d)}{\Lnot \val{\sigma}(s_1)}$;
\\
$\shlop{:}s_1{:}d$ & stands for $\iram{\OFunc}{q}$ where &
$q(\sigma) =
 \fncv{\sigma}{\reg{\sigma}(d)}{\shl \val{\sigma}(s_1)}$;
\\
$\shrop{:}s_1{:}d$ & stands for $\iram{\OFunc}{q}$ where &
$q(\sigma) =
 \fncv{\sigma}{\reg{\sigma}(d)}{\shr \val{\sigma}(s_1)}$;
\\
$\rolop{:}s_1{:}d$ & stands for $\iram{\OFunc}{q}$ where &
$q(\sigma) =
 \fncv{\sigma}{\reg{\sigma}(d)}{\rol \val{\sigma}(s_1)}$;
\\
$\rorop{:}s_1{:}d$ & stands for $\iram{\OFunc}{q}$ where &
$q(\sigma) =
 \fncv{\sigma}{\reg{\sigma}(d)}{\ror \val{\sigma}(s_1)}$;
\\
$\movop{:}s_1{:}d$ & stands for $\iram{\OFunc}{q}$ where &
$q(\sigma) =
 \fncv{\sigma}{\reg{\sigma}(d)}{\val{\sigma}(s_1)}$;
\\
$\eqop{:}s_1{:}s_2$ & stands for $\iram{p}{\IFunc}$ where &
$p(\sigma) = 1$ iff 
$\bton(\val{\sigma}(s_1)) = \bton(\val{\sigma}(s_2))$;
\\
$\gtop{:}s_1{:}s_2$ & stands for $\iram{p}{\IFunc}$ where &
$p(\sigma) = 1$ iff 
$\bton(\val{\sigma}(s_1)) > \bton(\val{\sigma}(s_2))$.
\end{tabular}
\end{ldispl}%

\section{Semi-Realistic RAM Programs}
\label{sect-SRRAM-programs}

In this section, we introduce a version of the RAM model of computation
that is intended to be a more or less realistic idealization of a real 
computer.
This version is obtained by restriction of the set of basic RAM 
instructions considered in \PGABTARAMI.

A \emph {semi-realistic RAM program}, called an \emph{SRRAM program} for 
short, is a closed \PGABTARAMI\ term of sort $\InSeq$ that is of the 
form $(t_1 \conc \ldots \conc t_n)\rep$, where each $t_i$ has one of the 
following forms:
\begin{ldispl}
\begin{tabular}[t]{@{}l@{\hsp{.75}}l@{\hsp{.6}}l @{}}
$a$                       & where &
$a \in \BInstrrram \diff \CInstrrram$;
\\
$\ptst{a} \conc \fjmp{l}$ & where & 
$a \in \CInstrrram$ \,and\, $l \in \Nat$;
\\
$\fjmp{l}$                & where & $l \in \Nat$;
\\
$\halt$.
\end{tabular}
\end{ldispl}%
In the \emph{SRRAM model of computation}, machines, called SRRAMs, 
consist of an SRRAM program together with a RAM memory on which it 
operates during execution.

A \emph {standard RAM program} is an SRRAM program in which only 
addition instructions, subtraction instructions, data transfer 
instructions, and comparison instructions occur (cf.~\cite{CR73a}).
A \emph {successor RAM program} is an SRRAM program in which only
addition instructions of the form $\addop{:}s_1{:}\#1{:}d$, data transfer 
instructions, and comparison instructions occur (cf.~\cite{Sch79a}).

The following theorem is a result concerning the computational power of
SRRAM programs.
\begin{theorem}
\label{theorem-Turing-computable}
For each $\pfunct{F}{{(\BitStr)}^n}{\BitStr}$, there exists an SRRAM 
program $P$ such that $P$ computes $F$ iff $F$ is Turing-computable.
\end{theorem}
\begin{proof}
The SRRAM model of computation is essentially the same as the MBRAM 
model of computation from~\cite{Emd90a} extended with shift/rotate 
instructions.
It follows directly from simulation results mentioned in~\cite{Emd90a} 
(part~(5) of Theorem~2.4, part~(1) of Theorem~2.5, and part~(3) of 
Theorem~2.6) that each MBRAM can be simulated by a Turing machine and 
vice versa.
Because each Turing machine can be simulated by a MBRAM, we immediate 
have that each Turing machine can be simulated by an SRRAM.
It is easy to see that the shift/rotate instructions can be simulated by 
a Turing machine.
From this and the fact that each MBRAM can be simulated by a Turing 
machine, it follows that each SRRAM can be simulated by a Turing machine
as well. 
Hence, each SRRAM is Turing equivalent to a Turing machine.
From this, the theorem follows immediately.
\qed
\end{proof}

Below, we write $\POLY$ for
$\set{T \where 
 \funct{T}{\Nat}{\Nat} \Land T \mathrm{\,is\,a\,polynomial\,function}}$.

The following theorem is a result relating the complexity class 
$\mathbf{PSPACE}$ to the functions from $\BitStr$ to $\Bit$ that can be 
computed by an SRRAM program in polynomial time.
\begin{theorem}
\label{theorem-PSPACE}
For each $\funct{F}{\BitStr}{\Bit}$, there exist an SRRAM program $P$ 
and a $T \in \POLY$ such that $P$ computes $F$ in time $T$ under the
uniform time measure iff $F \in \mathbf{PSPACE}$.
\end{theorem}
\begin{proof}
The SRRAM model of computation is essentially the same as the MRAM model 
of computation from~\cite{HS74a} extended with division and shift/rotate 
instructions.
We know from the main result of that paper that, for each 
$\funct{F}{\BitStr}{\Bit}$, 
there even exists an SRRAM program $P$ in which division and 
shift/rotate instructions do not occur and a $T \in \POLY$ such that $P$ 
computes $F$ in time $T$ under the uniform time measure iff 
$F \in \mathbf{PSPACE}$.
\qed
\end{proof}
Theorem~\ref{theorem-PSPACE} tell us that all decision problems that
belong to $\mathbf{PSPACE}$ can be solved by means of a SRRAM program in 
polynomial time.
This means that it is highly questionable whether the SRRAM model of 
computation is a reasonable model of computation.
However, it can be made a reasonable model by switching from the uniform
time measure to another time measure.
Such a time measure is introduced in Section~\ref{sect-bit-oriented}.

The proof of Theorem~\ref{theorem-PSPACE} reveals that the theorem still
holds if division and shift/rotate instructions are excluded from the 
SRRAM programs.
It turns out that we get another result if multiplication instructions 
are excluded as well.
\begin{theorem}
\label{theorem-P}
For each $\funct{F}{\BitStr}{\Bit}$, there exist an SRRAM program $P$
\linebreak[2] 
in which multiplication, division, and shift/rotate instructions do not 
occur and a $T \in \POLY$ such that $P$ computes $F$ in time $T$ under 
the uniform time measure iff $F \in \mathbf{P}$.
\end{theorem}
\begin{proof}
The model of computation obtained by excluding multiplication, division, 
and shift/rotate instructions from the SRRAM programs is the standard RAM 
model of computation extended with logical instructions.
From Theorem~2 in~\cite{CR73a}, we know that time complexity on standard 
RAMs under the uniform time measure and time complexity on multi-tape 
Turing machines are polynomially related.
It is easy to see that the logical instructions can be simulated by 
a multi-tape Turing machine in linear time.
Hence, the time complexities remain polynomially related if the standard 
RAM model is extended with logical instructions.
From this, the theorem follows immediately.
\qed
\end{proof}

\section{A Bit-Oriented Time Measure for SRRAM Programs}
\label{sect-bit-oriented}

In this section, we introduce a time measure for the SRRAM model of 
computation that has it origin in the idea that the time that it takes 
to execute an instruction on an SRRAM should be based on the number 
of steps that a multi-tape Turing machine with input alphabet $\Bit$ 
needs to simulate the instruction.
The choice have been made to make use of well-known upper bounds, but 
lesser upper bounds could have been used instead. 

We write $\mathcal{C}_\Thr$ for the set of all closed \PGABTARAMI\ terms 
of sort~$\Thr$.

We define a family $\boc{}$ of partial \emph{non-uniform cost} 
functions $\pfunct{\boc{\sigma}}{\mathcal{C}_\Thr}{\Nat}$, one for each 
$\sigma \in \MStates$, recursively as follows:
\begin{ldispl}
\begin{array}[t]{@{}l@{\quad}l@{}}
\boc{\sigma}(\Stop) = 0\;, & 
\\
\boc{\sigma}(\pcc{t}{\iram{p}{q}}{t'}) = 
\boc{\sigma}(\iram{p}{q}) + \boc{q(\sigma)}(t) & 
\mif p(\sigma) = \one\;, 
\\
\boc{\sigma}(\pcc{t}{\iram{p}{q}}{t'}) = 
\boc{\sigma}(\iram{p}{q}) + \boc{q(\sigma)}(t') & 
\mif p(\sigma) = \zero\;, 
\end{array}
\end{ldispl}%
where the family of partial functions 
$\pfunct{\boc{\sigma}}{\BInstrram}{\Nat}$ (defined for all basic RAM 
instructions from $\BInstrrram$), one for each $\sigma \in \MStates$, is 
defined as follows: 
\begin{ldispl}
\begin{array}[t]{@{}l@{\quad}l@{}}
\boc{\sigma}(\binop{:}s_1{:}s_2{:}d) = 
\max \set{\boc{\sigma}(s_1),\boc{\sigma}(s_2)} + \bocp{\sigma}(d)  
 & \mif \binop \notin \set{\mulop,\divop}\;,
\\
\boc{\sigma}(\binop{:}s_1{:}s_2{:}d) = 
\boc{\sigma}(s_1) \mul \boc{\sigma}(s_2) + \bocp{\sigma}(d)
 & \mif \binop \in \set{\mulop,\divop}\;,
\\
\boc{\sigma}(\unop{:}s_1{:}d) = \boc{\sigma}(s_1) + \bocp{\sigma}(d)\;, 
\\
\boc{\sigma}(\cmpop{:}s_1{:}s_2) = 
\max \set{\boc{\sigma}(s_1),\boc{\sigma}(s_2)}\;,
\end{array}
\end{ldispl}%
where the family of total functions $\funct{\boc{\sigma}}{\Src}{\Nat}$, 
one for each $\sigma \in \MStates$, is defined as follows:%
\footnote
{Recall that $\len(w)$, where $w \in \Bit^*$, stands for the length of
 $w$.}
\begin{ldispl}
\boc{\sigma}(\# i) = \len(\ntob(i))\;,
\\
\boc{\sigma}(i)    = \len(\ntob(i)) + \len(\val{\sigma}(i))\;,
\\
\boc{\sigma}(@ i)  = \len(\ntob(i)) + \len(\val{\sigma}(i)) + 
                     \len(\val{\sigma}(\bton(\val{\sigma}(i))))
\end{ldispl}%
and the family of total functions $\funct{\bocp{\sigma}}{\Dst}{\Nat}$, 
one for each $\sigma \in \MStates$, is defined as follows:
\begin{ldispl}
\bocp{\sigma}(i)    = \len(\ntob(i))\;,
\\
\bocp{\sigma}(@ i)  = \len(\ntob(i)) + \len(\val{\sigma}(i))\;.
\end{ldispl}%

Let $t$ be a closed \PGABTARAMI\ term of sort $\InSeq$,
let $n \in \Nat$, let $\pfunct{F}{{(\BitStr)}^n}{\BitStr}$, and
let $\funct{T}{\Nat}{\Nat}$.
Then $t$ \emph{computes $F$ in time $T$ under the bit-oriented time 
measure} if:
\begin{itemize}
\item
for all $w_1,\ldots,w_n \in \BitStr$ such that $F(w_1,\ldots,w_n)$ 
is defined, there exist a $\sigma \in \MStates$ such that: 
\begin{ldispl}
\extr{t} \sfapply \sram{\fncvn{\sigma_\emptystr}{1}{w_1}{n}{w_n}} = 
\sram{\fncv{\sigma}{0}{F(w_1,\ldots,w_n)}}\;,
\seqnsep 
\boc{\fncvn{\sigma_\emptystr}{1}{w_1}{n}{w_n}}(\extr{t})
 \leq
T(\len(w_1) + \ldots + \len(w_n))\;; 
\end{ldispl}%
\item
for all $w_1,\ldots,w_n \in \BitStr$ such that $F(w_1,\ldots,w_n)$ is
undefined:
\begin{ldispl}
\extr{t} \sfapply \sram{\fncvn{\sigma_\emptystr}{1}{w_1}{n}{w_n}} = 
\sram{\Div}\;.
\end{ldispl}%
\end{itemize}
Fine-tuning this definition boils down to adapting the definition of the
family of partial functions $\pfunct{\boc{\sigma}}{\BInstrram}{\Nat}$.

The parts of Theorem~2 from~\cite{CR73a} that concern standard RAMs under 
the logarithmic time measure hold also for SRRAMs under the bit-oriented 
time measure.  
\begin{theorem}
\label{theorem-quadratic-simuls}
For each $\pfunct{F}{{(\BitStr)}^n}{\BitStr}$: 
\begin{itemize}
\item[$(a)$]
if there exist an SRRAM program $P$ and a $\funct{T}{\Nat}{\Nat}$ such 
that $P$ computes \linebreak[2] $F$ in time $T$ under the bit-oriented time measure,
then there exists a multi-tape Turing machine $M$ such that $M$ computes 
$F$ in time $O(T^2)$;
\item[$(b)$]
if there exist a multi-tape Turing machine $M$ and a 
$\funct{T}{\Nat}{\Nat}$ such that \linebreak[2] $M$ computes $F$ in time $T$, then 
there exists an SRRAM program $P$ such that $P$ computes $F$ in time 
$O(T \mul \log_2(T))$ under the bit-oriented time measure. 
\end{itemize}
\end{theorem}
\begin{proof}
In the proof of (a), one of the working tapes of $M$ is considered 
to hold a representation of the state of the RAM memory on which $P$ 
operates during execution.
A RAM memory state is represented on this working tape by a string over 
the alphabet $\set{0,1,\Blank}$ that belongs to the set defined by the 
regular expression 
$(\Blank \Blank (0 + 1) (0 + 1)^* \Blank (0 + 1) (0 + 1)^*)^*$.  
The working tape content 
$\Blank \Blank w_1 \Blank w'_1 \ldots \Blank \Blank w_n \Blank w'_n
 \Blank \Blank \Blank \ldots$ 
represents the RAM memory state
$\fncvn{\sigma_\emptystr}{\bton(w_1)}{w'_1}{\bton(w_n)}{w'_n}$.

Take arbitrary $\sigma \in \MStates$ and $i \in \Nat$, and 
let $w = \ntob(i)$ and $w' = \val{\sigma}(i)$.
Then $\Blank \Blank w \Blank w'$ occurs in the 
representation of $\sigma$ iff $\sigma(i) \neq \emptystr$.
Moreover, 
$\len(\Blank \Blank w \Blank w') =
 \len(\ntob(i)) + \len(\val{\sigma}(i))$.
From this and the definition of the function
$\pfunct{\boc{\sigma}}{\BInstrram}{\Nat}$, it follows readily that, 
if $u \in \BInstrrram$ and $u$ is of the form
$\binop{:}s_1{:}s_2{:}i$, $\unop{:}s_1{:}i$, $\binop{:}s_1{:}s_2{:}@ j$ 
or $\unop{:}s_1{:}@ j$, where $\val{\sigma}(j) = i$, then
$\len(\Blank \Blank w \Blank w')$ is bounded by a constant times 
$\boc{\sigma}(u)$. 
From this and the fact that $P$ computes $F$ in time $T$ under the 
bit-oriented time measure, it follows immediately that the length of the 
representation of $\sigma$ on the work tape is bounded by $O(T)$.
This means that searching the working tape for the entry of a register 
takes at most $O(T)$ steps.
Since $\boc{\sigma}(u) \geq 1$ for all instructions $u \in \BInstrrram$, 
at most $O(T)$ instructions are executed.
Because each instruction $u \in \BInstrrram$ involves a constant number
of searches for register entries on the working tape, this means that 
the total number of steps spent on searching the working tape for the 
entries of registers is bounded by $O(T^2)$.

The functions $\pfunct{\boc{\sigma}}{\BInstrram}{\Nat}$ are defined such 
that, in the case that $P$ computes $F$ in time $T$ under the 
bit-oriented time measure, the total number of steps that a multi-tape 
Turing machine needs to compute $F$, not counting the steps spent on 
searching the working tape for the entries of registers, is bounded by 
$O(T)$.
Because the total number of steps spent on searching the working tape 
for the entries of registers is bounded by $O(T^2)$, this means that the 
total number of steps that a multi-tape Turing machine needs to compute 
$F$ is bounded by $O(T^2)$.

In the proof of (b), $M$ is assumed to have $k$ tapes. 
The state of the RAM memory on which $P$ operates during execution is 
considered to represent the contents of the $k$ tapes of $M$ as follows: 
the content of the $i$th cell on the $j$th tape is the content of the 
register with number $k \mul i + j + c$, where $c$ is the number of 
auxiliary registers that $P$ needs to simulate Turing machine steps.
The auxiliary registers include $k$ registers for the positions of the 
$k$ tape heads.
$P$ can read out or alter the cells under the $k$ tape heads by means 
of indirect addressing through these position-holding registers.

It follows immediately from the definition of the functions 
$\pfunct{\boc{\sigma}}{\BInstrram}{\Nat}$ that, under the bit oriented 
time measure, the time that an SRRAM program needs per Turing machine 
step is a constant plus the time spent on accessing the registers that 
contain the contents of the cells under the tape heads.
In the case that $M$ computes $F$ in time $T$, the number of cells $M$ 
can move tape heads away from the starting position is bounded by $T$.
From this, the fact that an SRRAM program uses indirect addressing to 
access the registers that contain the contents of the cells under the 
tape heads, and the fact that 
$\len(\ntob(i)) = \lfloor \log_2(i) \rfloor + 1$ if $i > 0$,
it follows immediately that the time it takes an SRRAM program computing 
$F$ to access these farthest tape cells is bounded by $O(\log_2(T))$.
Because each step of $M$ involves accessing the cells under its tape 
heads, this means that the time that a SRRAM program needs to compute 
$F$ is bounded by $O(T \mul \log_2(T))$.
\qed
\end{proof}
The approaches to the proofs of the two parts of 
Theorem~\ref{theorem-quadratic-simuls} have been inspired by the proofs 
of the corresponding parts of Theorem~2 from~\cite{CR73a}. 

The following corollary of Theorem~\ref{theorem-quadratic-simuls} is a 
counterpart of Theorem~\ref{theorem-P}.
\begin{corollary}
\label{theorem-P-ext}
For each $\funct{F}{{(\BitStr)}^n}{\Bit}$, there exist an SRRAM program 
$P$ and a $T \in \POLY$ such that $P$ computes $F$ in time $T$ under the 
bit-oriented time measure iff $F \in \mathbf{P}$.
\end{corollary}

\section{A Bit-Oriented Space Measure for SRRAM Programs}
\label{sect-bit-oriented-space}

In this section, we introduce for the sake of completeness a 
bit-oriented space measure for the SRRAM model of computation.
This space measure originates from~\cite{Emd90a}.

Let $t$ be a closed \PGABTARAMI\ term of sort $\InSeq$,
let $n \in \Nat$, let $\pfunct{F}{{(\BitStr)}^n}{\BitStr}$, and 
let $\funct{S}{\Nat}{\Nat}$.
Then $t$ \emph{computes $F$ in space $S$} if $t$ computes $F$ and,
for all $w_1,\ldots,w_n \in \BitStr$ such that $F(w_1,\ldots,w_n)$ 
is defined, for some $m \in \Natpos$, there exist closed \PGABTARAMI\ 
terms $t_1,\ldots,t_m$ of sort $\Thr$ and RAM memory states 
$\sigma_1,\ldots,\sigma_m$ such that:
\begin{itemize}
\item 
$t_1 = \extr{t}$;
\item
$t_m = \Stop$;
\item
$\sigma_1(i) = \emptystr$ 
for all $i \in \Nat$ with $i \notin \set{1,\ldots,n}$; 
\item
$\sigma_j(i) = w_i$ 
for all $i \in \set{1,\ldots,n}$ and $j \in \set{1,\ldots,m}$;
\item
$\sigma_m(0) = F(w_1,\ldots,w_n)$;
\item
for all $j \in \set{1,\ldots,m}$, 
$t_j \sfapply \sram{\sigma_j} = t_{j+1} \sfapply \sram{\sigma_{j+1}}$
is a closed substitution instance of an instance of axiom schema A4 or 
A5;
\item
$\max 
 \set{\sum_{i \in \Nat \diff \set{1,\ldots,n}}
       (\len(i) + \len(\sigma_j(i))) \where
      j \in \set{1,\ldots,m}} \leq
 S(\len(w_1) + \ldots + \len(w_n))$.
\end{itemize}
The pairs $(t_j,\sigma_j)$, for $j \in \set{1,\ldots,m}$, can be 
looked upon as SRRAM configurations and the sequence  
$(t_1,\sigma_1)\cat \ldots \cat (t_m,\sigma_m)$ 
can be looked upon as a SRRAM computation.
Instead of introducing off-line SRRAMs, we require that during 
computations the contents of the input registers are never changed.
 
In the above definition space is essentially measured following the 
third method mentioned in~\cite{Emd90a}, using the function 
$\mathit{size_b}$ from that paper as size function.
By this space measure, it is guaranteed that space complexity on SRRAMs 
and space complexity on multi-tape Turing machines are related by a 
constant factor.

\section{Discussion on the Bit-Oriented Time Measure}
\label{sect-discussion}

In the field of computational complexity, a model of computation is 
considered a reasonable sequential model of computation if time 
complexity on its machines and time complexity on multi-tape Turing 
machines are polynomially related and space complexity on its machines 
and space complexity on multi-tape Turing machines are related by a 
constant factor (cf.\ the Invariance Thesis in~\cite{Emd90a}).
This makes the complexity classes that represent the fundamental 
concepts of computational complexity, i.e.\ L, NL, P, NP, PSPACE, 
NPSPACE, EXP, NEXP, EXPSPACE, NEXPSPACE, machine-independent insofar as 
reasonable sequential models of computation are concerned.

The logarithmic time measure has been introduced in all but the simplest 
known versions of the RAM model of computation to obtain a reasonable 
model.
However, it is questionable whether the logarithmic time measure is the 
most natural time measure.
It takes the lengths of the bit strings involved in the execution of an
instruction into account, but not the operation involved.
The logarithmic time measure works in the case of the known versions of 
the RAM model of computation only because the operations involved can 
always be simulated by a multi-tape Turing machine in polynomial time.

The bit-oriented time measure introduced in this paper takes both the 
operation and the lengths of the bit strings involved in the execution 
of an instruction into account.
Thereby, the bit-oriented time measure actually takes the total number 
of operations on bits involved in the execution of an instruction into 
account.
This property is arguably the best justification of a time measure 
intended to make the time measures of different models of computation 
comparable.

With $\Bit$ as input alphabet, a version of the Turing machine model of 
computation supports operations on bits more directly than most other 
well-known models of computation.
This has been an important reason to consider in this paper the times 
that it takes to carry out the operations on bits in the setting of a 
version of the Turing machine model.
Another important reason has been that the complexity classes that 
represent the fundamental concepts of computational complexity were 
initially introduced and studied in the setting of the multi-tape Turing 
machine model.

The extended logarithmic time measure introduced in~\cite{Die11a} also
takes the total number of operations on bits involved in the execution 
of an instruction into account, but, there, the choice is made to 
consider the times that it takes to carry out the operations on bits in 
the setting of the successor RAM model.
This is the most primitive version of the RAM model and supports 
operations on bits equally directly as multi-tape Turing machine model.
The approach of~\cite{Die11a} may be advantageous if one is interested 
in relating complexity results based on other versions of the RAM model 
to complexity results based on the successor RAM model, but is 
disadvantageous if one is interested in relating complexity results 
based on versions of the RAM model to complexity results based on the 
multi-tape Turing machine model.

The idea behind the bit-oriented time measure from this paper is that 
the time that it takes to execute an instruction on an SRRAM should be 
based on the number of steps that a multi-tape Turing machine with input 
alphabet $\Bit$ needs to simulate the instruction.
Moreover, the choice has been made to use well-known polynomial upper 
bounds. 
By producing the bit-oriented measure in this way for programs of RAMs 
of a kind obtained by restricting the set of basic RAM instructions of 
\PGABTARAMI\ in another way than for SRRAM programs, it is guaranteed
that Theorem~\ref{theorem-quadratic-simuls} holds for these programs as 
well.  
Examination of the proof of that theorem teaches us that it depends only
on the use of upper bounds for the number of steps that a multi-tape 
Turing machine with input alphabet $\Bit$ needs to simulate the 
instructions that may occur in the programs.

\section{Concluding Remarks}
\label{sect-concl}

We have presented an instantiation of a parameterized algebraic theory 
of single-pass instruction sequences, the behaviours produced by such 
instruction sequences under execution, and the interaction between such 
behaviours and components of an execution environment for instruction
sequences. 
In the instantiation concerned, RAM memories are taken as the components 
of an execution environment, instructions for a RAM are taken as basic 
instructions, and an execution environment consists of only one 
component. 
Because we have opted for the most general instantiation, all 
instructions that do not read out or alter more than one register from 
the RAM memory are taken as basic instructions.

The presentation of the instantiation has been set up in such a way that 
the introduction of services, the generic kind of execution-environment 
components from the parameterized theory, is circumvented.
In~\cite{BM18b}, the presentation of another instantiation of the same 
parameterized theory has been set up in the same way.
The distinguishing feature of this way of presenting an instantiation of 
the parameterized theory is that it yields a less involved presentation 
than the way adopted in earlier work based on an instantiation of this 
parameterized \linebreak[2] theory.

We have provided evidence for the claim that the presented algebraic 
theory provides a setting for the development of theory in areas such as 
computational complexity and analysis of algorithms that is more general 
than the setting provided by some known version of the RAM model of 
computation.
We have among other things shown that a relatively unknown, but 
realistic, version of the RAM model can be dealt with in the setting 
concerned by imposing apposite restrictions.
For this model, an alternative to the usual time measures for versions 
of the RAM model, called the bit-oriented time measure has been 
introduced.

Related to the introduction of the bit-oriented time measure is the 
choice for registers that contain bit strings instead of natural 
numbers.
Whereas it is usual in versions of the RAM model of computation that bit 
strings are represented by natural numbers, here natural numbers are 
represented by bit strings. 
Moreover, the choice has been made to represent the natural number $0$ 
by the bit string $0$ and to adopt the empty bit string as the register 
content that indicates that a register is (as yet) unused.
Therefore, we have, as in most other versions of the RAM model, 
$\len(0) = 1$ and $\len(i+1) = \lfloor \log_2(i + 1) \rfloor + 1$ 
if $\len$ on natural numbers is simply defined by 
$\len(i) = \len(\ntob(i))$ (as before $\len(w)$, where $w \in \Bit^*$, 
stands for the length of $w$).
  
The closed terms of the presented algebraic theory that are used as RAM
programs can be considered to constitute a programming language of which 
the syntax and semantics is defined following an algebraic approach.
However, this approach is more operational than the usual algebraic 
approach, which is among other things followed 
in~\cite{BDMW81a,BWP87a,GM96a}.
The more operational approach is advantageous in the case of a language 
that is used to investigate issues in the areas of computational 
complexity and analysis of algorithms.

The work presented in this paper is among other things concerned with 
formalization in the areas of computational complexity and analysis of 
algorithms.
To the best of my knowledge, very little work has been done in this 
area.
Three notable exceptions are~\cite{Nor11a,XZU13a,AR15a}.
However, those papers are concerned with formalization in a theorem
prover (HOL4, Isabelle/HOL, Matita) and focussed on some version of the 
Turing machine model of computation.
This makes it impracticable to compare the work presented in those 
papers with the work presented here.

The contributions of this paper to the work on models of computation 
rely heavily on~\cite{CR73a,HS74a}.
A variant of the bit-oriented time measure has been proposed 
in~\cite{Die11a}. 

This paper introduces a setting for the development of theory in areas 
such as computational complexity and analysis of algorithms using 
virtually any version of the RAM model of computation.
This setting is an instantiation of a parametrized algebraic theory. 
Several other models of computation can be covered by other 
instantiations of this theory.
The instantiation for the Turing machine model of computation is
described in~\cite{BM19a}.
However, the theory concerned is not general enough to cover parallel 
models of computation.
An interesting option for future work is to study how it can be 
extended to a theory that covers parallel models of computation.

\bibliographystyle{plain}

\bibliography{IS}

\begin{thebibliography}{10}

\bibitem{AHU74a}
A.~V. Aho, J.~E. Hopcroft, and J.~D. Ullman.
\newblock {\em The Design and Analysis of Computer Algorithms}.
\newblock Addison-Wesley, Reading, MA, 1974.

\bibitem{AR15a}
A.~Asperti and W.~Ricciotti.
\newblock A formalization of multi-tape {Turing} machines.
\newblock {\em Theoretical Computer Science}, 603:23--42, 2015.
\newblock \doi{10.1016/j.tcs.2015.07.013}

\bibitem{BW90}
J.~C.~M. Baeten and W.~P. Weijland.
\newblock {\em Process Algebra}, volume~18 of {\em Cambridge Tracts in
  Theoretical Computer Science}.
\newblock Cambridge University Press, Cambridge, 1990.
\newblock \doi{10.1017/CBO9780511624193}

\bibitem{BL02a}
J.~A. Bergstra and M.~E. Loots.
\newblock Program algebra for sequential code.
\newblock {\em Journal of Logic and Algebraic Programming}, 51(2):125--156,
  2002.
\newblock \doi{10.1016/S1567-8326(02)00018-8}

\bibitem{BM09g}
J.~A. Bergstra and C.~A. Middelburg.
\newblock Inversive meadows and divisive meadows.
\newblock {\em Journal of Applied Logic}, 9(3):203--220, 2011.
\newblock \doi{10.1016/j.jal.2011.03.001}

\bibitem{BM09k}
J.~A. Bergstra and C.~A. Middelburg.
\newblock Instruction sequence processing operators.
\newblock {\em Acta Informatica}, 49(3):139--172, 2012.
\newblock \doi{10.1007/s00236-012-0154-2}

\bibitem{BM12b}
J.~A. Bergstra and C.~A. Middelburg.
\newblock {\em Instruction Sequences for Computer Science}, volume~2 of {\em
  Atlantis Studies in Computing}.
\newblock Atlantis Press, Amsterdam, 2012.
\newblock \doi{10.2991/978-94-91216-65-7}

\bibitem{BM13b}
J.~A. Bergstra and C.~A. Middelburg.
\newblock Instruction sequence expressions for the secure hash algorithm
  {SHA-256}, 2013.
\newblock \arXiv{1308.0219}

\bibitem{BM13a}
J.~A. Bergstra and C.~A. Middelburg.
\newblock Instruction sequence based non-uniform complexity classes.
\newblock {\em Scientific Annals of Computer Science}, 24(1):47--89, 2014.
\newblock \doi{10.7561/sacs.2014.1.47}

\bibitem{BM14a}
J.~A. Bergstra and C.~A. Middelburg.
\newblock On algorithmic equivalence of instruction sequences for computing bit
  string functions.
\newblock {\em Fundamenta Informaticae}, 138(4):411--434, 2015.
\newblock \doi{10.3233/fi-2015-1219}

\bibitem{BM14e}
J.~A. Bergstra and C.~A. Middelburg.
\newblock Instruction sequence size complexity of parity.
\newblock {\em Fundamenta Informaticae}, 149(3):297--309, 2016.
\newblock \doi{10.3233/FI-2016-1450}

\bibitem{BM13c}
J.~A. Bergstra and C.~A. Middelburg.
\newblock Instruction sequences expressing multiplication algorithms.
\newblock {\em Scientific Annals of Computer Science}, 28(1):39--66, 2018.
\newblock \doi{10.7561/sacs.2018.1.39}

\bibitem{BM18b}
J.~A. Bergstra and C.~A. Middelburg.
\newblock A short introduction to program algebra with instructions for
  {Boolean} registers.
\newblock {\em Computer Science Journal of Moldova}, 26(3):199--232, 2018.

\bibitem{BM19a}
J.~A. Bergstra and C.~A. Middelburg.
\newblock Program algebra for {Turing}-machine programs.
\newblock {\em Scientific Annals of Computer Science}, 29(2):113--139, 2019.
\newblock \doi{10.7561/SACS.2019.2.113}

\bibitem{BM18a}
J.~A. Bergstra and C.~A. Middelburg.
\newblock On the complexity of the correctness problem for non-zeroness test
  instruction sequences.
\newblock {\em Theoretical Computer Science}, 802:1--18, 2020.
\newblock \doi{10.1016/j.tcs.2019.03.040}

\bibitem{BT07a}
J.~A. Bergstra and J.~V. Tucker.
\newblock The rational numbers as an abstract data type.
\newblock {\em Journal of the ACM}, 54(2):Article 7, 2007.
\newblock \doi{10.1145/1219092.1219095}

\bibitem{BDMW81a}
M.~Broy, W.~Dosch, B.~M{\"o}ller, and M.~Wirsing.
\newblock {GOTOs} -- a study in the algebraic specification of programming
  languages (extended abstract).
\newblock In W.~Brauwer, editor, {\em GI --- 11. Jahrestagung}, volume~50 of
  {\em Informatik-Fachberichte}, pages 109--121. Springer-Verlag, 1981.
\newblock \doi{10.1007/978-3-662-01089-1_13}

\bibitem{BWP87a}
M.~Broy, M.~Wirsing, and P.~Pepper.
\newblock On the algebraic definition of programming languages.
\newblock {\em ACM Transactions on Programming Languages and Systems},
  9(1):54--99, 1987.
\newblock \doi{10.1145/9758.10501}

\bibitem{CR73a}
S.~A. Cook and R.~A. Reckhow.
\newblock Time bounded random access machine.
\newblock {\em Journal of Computer and System Sciences}, 7(4):354--375, 1973.
\newblock \doi{10.1016/S0022-0000(73)80029-7}

\bibitem{Die11a}
C.~Diem.
\newblock On the notion of bit complexity.
\newblock {\em Bulletin of the {EATCS}}, 103:36--52, 2011.

\bibitem{EM85a}
H.~Ehrig and B.~Mahr.
\newblock {\em Fundamentals of Algebraic Specification {I}: Equations and
  Initial Semantics}, volume~6 of {\em EATCS Monographs}.
\newblock Springer-Verlag, Berlin, 1985.
\newblock \doi{10.1007/978-3-642-69962-7}

\bibitem{Gog21a}
J.~A. Goguen.
\newblock Theorem proving and algebra, 2021.
\newblock \arXiv{2101.02690}

\bibitem{GM96a}
J.~A. Goguen and G.~Malcolm.
\newblock {\em Algebraic Semantics of Imperative Programs}.
\newblock Foundations of Computing. MIT Press, Cambridge, MA, 1996.
\newblock \doi{10.7551/mitpress/1188.001.0001}

\bibitem{HS74a}
J.~Hartmanis and J.~Simon.
\newblock On the power of multiplication in random access machines.
\newblock In {\em SWAT '74}, pages 13--23. IEEE, 1974.
\newblock \doi{10.1109/SWAT.1974.20}

\bibitem{SiteIS}
C.~A. Middelburg.
\newblock Instruction sequences as a theme in computer science, 2021.
\newblock \url{https://instructionsequence.wordpress.com/}

\bibitem{Mor98a}
B.~M. Moret.
\newblock {\em The Theory of Computation}.
\newblock Addison-Wesley, Reading, MA, 1998.

\bibitem{Nor11a}
M.~Norrish.
\newblock Mechanised computability theory.
\newblock In M.~van Eekelen, H.~Geuvers, J.~Schmaltz, and F.~Wiedijk, editors,
  {\em ITP 2011}, volume 6898 of {\em Lecture Notes in Computer Science}, pages
  297--311. Springer-Verlag, 2011.
\newblock \doi{10.1007/978-3-642-22863-6_22}

\bibitem{Pap94a}
C.~H. Papadimitriou.
\newblock {\em Computational Complexity}.
\newblock Addison-Wesley, Reading, MA, 1994.

\bibitem{ST12a}
D.~Sannella and A.~Tarlecki.
\newblock {\em Foundations of Algebraic Specification and Formal Software
  Development}.
\newblock Monographs in Theoretical Computer Science, An EATCS Series.
  Springer-Verlag, Berlin, 2012.
\newblock \doi{10.1007/978-3-642-17336-3}

\bibitem{Sch79a}
A.~Sch{\"o}nhage.
\newblock On the power of random access machines.
\newblock In H.~A. Maurer, editor, {\em ICALP'79}, volume~71 of {\em Lecture
  Notes in Computer Science}, pages 520--529. Springer-Verlag, 1979.
\newblock \doi{10.1007/3-540-09510-1_42}

\bibitem{Emd90a}
P.~van Emde~Boas.
\newblock Machine models and simulations.
\newblock In J.~van Leeuwen, editor, {\em Handbook of Theoretical Computer
  Science}, volume~A, pages 2--66. Elsevier, Amsterdam, 1990.
\newblock \doi{10.1016/B978-0-444-88071-0.50006-0}

\bibitem{GV93}
R.~J. van Glabbeek and F.~W. Vaandrager.
\newblock Modular specification of process algebras.
\newblock {\em Theoretical Computer Science}, 113(2):293--348, 1993.
\newblock \doi{10.1016/0304-3975(93)90006-F}

\bibitem{Wir90a}
M.~Wirsing.
\newblock Algebraic specification.
\newblock In J.~van Leeuwen, editor, {\em Handbook of Theoretical Computer
  Science}, volume~B, pages 675--788. Elsevier, Amsterdam, 1990.
\newblock \doi{10.1016/B978-0-444-88074-1.50018-4}

\bibitem{XZU13a}
J.~Xu, X.~Zhang, and C.~Urban.
\newblock Mechanising {Turing} machines and computability theory in
  {Isabelle/HOL}.
\newblock In S.~Blazy, C.~Paulin-Mohring, and D.~Pichardie, editors, {\em ITP
  2013}, volume 7998 of {\em Lecture Notes in Computer Science}, pages
  147--162. Springer-Verlag, 2013.
\newblock \doi{10.1007/978-3-642-39634-2_13}

\end{thebibliography}

\end{document}